\documentclass[conference]{IEEEtran}%
\usepackage{amsfonts}
\usepackage{amsmath}
\usepackage{amssymb}
\usepackage{graphicx}
\usepackage{cite}%
\newtheorem{theorem}{Theorem}

\newtheorem{corollary}[theorem]{Corollary}

\newtheorem{definition}[theorem]{Definition}

\begin{document}

\title{Utility and Privacy of Data Sources: Can Shannon Help Conceal and Reveal Information?}
\pubid{~}
\specialpapernotice{~}%

%TCIMACRO{\TeXButton{Author Information}{\author{\authorblockN
%{Lalitha Sankar\authorrefmark{1}, S. Raj Rajagopalan \authorrefmark
%{2}, H. Vincent Poor\authorrefmark{1}}
%\authorblockA{\authorrefmark{1}Dept. of Electrical Engineering,
%Princeton University,
%Princeton, NJ 08544.
%{lalitha,poor}@princeton.edu\\}
%\authorblockA{\authorrefmark{2}HP Labs,
%Princeton, NJ 08544.
%raj.raj@hp.com\\}}}}%
%BeginExpansion
\author{\authorblockN{Lalitha Sankar\authorrefmark{1}%
, S. Raj Rajagopalan \authorrefmark{2}, H. Vincent Poor\authorrefmark{1}}
\authorblockA{\authorrefmark{1}Dept. of Electrical Engineering,
Princeton University,
Princeton, NJ 08544.
{lalitha,poor}@princeton.edu\\}
\authorblockA{\authorrefmark{2}HP Labs,
Princeton, NJ 08544.
raj.raj@hp.com\\}}%
%EndExpansion
%

%TCIMACRO{\TeXButton{Make Title}{\maketitle}}%
%BeginExpansion
\maketitle
%EndExpansion
%

%TCIMACRO{\TeXButton{Begin abstract}{\begin{abstract}}}%
%BeginExpansion
\begin{abstract}%
%EndExpansion

\footnotetext{This research is supported in part by the National\ Science
Foundation under Grants CNS-09-05086 and CNS-09-05398.}The problem of private
information \textquotedblleft leakage\textquotedblright\ (inadvertently or by
malicious design) from the myriad large centralized searchable data
repositories drives the need for an analytical framework that quantifies
unequivocally how safe private data can be (privacy) while still providing
useful benefit (utility) to multiple legitimate information consumers. Rate
distortion theory is shown to be a natural choice to develop such a framework
which includes the following: modeling of data sources, developing application
independent utility and privacy metrics, quantifying utility-privacy tradeoffs
irrespective of the type of data sources or the methods of providing privacy,
developing a side-information model for dealing with questions of external
knowledge, and studying a \textit{successive disclosure} problem for multiple
query data sources.%

%TCIMACRO{\TeXButton{End abstract}{\end{abstract}}}%
%BeginExpansion
\end{abstract}%
%EndExpansion

\section{Introduction}

Information technology and electronic communications have been rapidly applied
to almost every sphere of human activity, including commerce, medicine and
social networking. The concomitant emergence of myriad large centralized
searchable data repositories has made \textquotedblleft
leakage\textquotedblright\ of private information such as medical data, credit
card information, or social security numbers via data correlation
(inadvertently or by malicious design) an important and urgent societal
problem. Unlike the well-studied secrecy problem (e.g.,
\cite{RSA,CsisNar1,Wyner}) in which the protocols or primitives make a sharp
distinction between secret and non-secret data, in the \emph{privacy} problem,
disclosing data provides informational utility while enabling possible loss of
privacy at the same time. In fact, in the course of a legitimate transaction,
a user can learn some public information, which is allowed and needs to be
supported, and at the same time also learn/infer private information, which
needs to be prevented. Thus every user is (potentially)\ also an adversary.
This drives the need for a unified analytical framework that can tell us
unequivocally and precisely how safe private data can be (privacy) while still
providing useful benefit (utility) to multiple legitimate information consumers.

It has been noted that utility and privacy are competing goals:
\textit{perfect privacy can be achieved by publishing nothing at all, but this
has no utility; perfect utility can be obtained by publishing the data exactly
as received, but this offers no privacy }\cite{Chawla01}. Utility of a data
source is potentially (but not necessarily) degraded when it is restricted or
modified to uphold privacy requirements. The central problem of this paper is
a precise quantification of the tradeoff between the privacy needs of the
\textit{respondents} (individuals represented by the data)\ and the utility of
the \textit{sanitized} (published) data for any data source.

Though the problem of privacy and information leakage has been studied for
several decades by multiple research communities (e.g.,
\cite{Adam_Wort,Dalenius,Sweeney,Ag_Ag,Chawla01} and the references therein),
the proposed solutions have been both heuristic and application-specific. The
recent groundbreaking theory of differential privacy
\cite{Dwork_DP,Dwork_DP_Survey} from the theoretical computer science
community is the first universal model that applies to any statistical
database irrespective of application or content. However, the crucial
challenges of an analytic characterization of both a utility metric and the
privacy-utility tradeoff remains unaddressed. We seek to address these
challenges using tools and techniques from information theory.

Rate distortion theory is a natural choice to study the utility-privacy
tradeoff; utility can be quantified via fidelity which in turn is related to
\textit{distortion} and privacy can be quantified via \textit{equivocation}.
Our key insight is captured in the following theorem we present in this paper:
for a data source with private and public data and desired utility level,
maximum privacy for the private data is achieved by \textit{minimizing the
information disclosure rate }sufficient to satisfy the desired utility for the
public data. To the best of our knowledge this is the first observation that
tightly relates utility and privacy. I

In a sparsely referenced paper \cite{Yamamoto} from three decades ago,
Yamamoto developed the tradeoff between rate, distortion, and equivocation for
a specific and simple source model. In this paper, we show via the above
summarized theorem that Yamamoto's formalism can be translated into the
language of data disclosure. Furthermore, we develop a framework which allows
us to model data sources, specifically databases, develop application
independent utility and privacy metrics, quantify the fundamental bounds on
the utility-privacy tradeoffs, and develop a side-information model for
dealing with questions of external knowledge, and study the utility-privacy
tradeoffs for multiple query data sources as a \textit{successive disclosure
}problem. The final problem arises in the following context: real-world data
sources are in general interactive, that is, they allow users multiple
interactions (queries). However, modeling this analytically is particularly
challenging. Our framework can handle the non-interactive (single query) case
for any given utility and privacy requirements. In this paper, we study the
interactive case as a \textit{successive disclosure problem} modeled along the
information-theoretic successive refinement problem.

\section{The Database Privacy Problem}

\subsection{Problem Definition}

While the problem of quantifying the utility/privacy problem applies to all
types of data sources, we start our study with databases because they are
highly structured and historically better studied than other types of sources.
A database is a table (matrix) whose rows represent the individual entries and
whose columns represent the \emph{attributes} of each entry \cite{Adam_Wort}.
For example, the attributes of each entry in a healthcare database typically
include name, address, social security number (SSN), gender, and a collection
of medical information and each entry contains the information pertaining to
an individual. Messages from a \emph{user} to a database are called
\emph{queries} and, in general, result in some numeric or non-numeric
information from the database termed the \emph{response}.

The goal of privacy protection is to ensure that, to the extent possible, the
user's knowledge is not increased beyond strict predefined limits by
interacting with the database. The goal of utility provision is, generally, to
maximize the amount of information that the user can receive. Depending on the
relationship between attributes, and the distribution of the actual data, a
response may contain information that can be inferred beyond what is
explicitly included in the response. The privacy policy defines the
information that should not revealed explicitly or by inference to the user
and depends on the context and the application. For example, in a database on
health statistics, attributes such as name and SSN may be considered private
data, whereas in a state motor vehicles\ database only the SSN\ is considered
private. The challenge for privacy protection is to design databases such that
any response does not reveal information contravening the privacy policy.

\subsection{Current Approaches and Metrics}

The approaches considered in the literature have centered on
\emph{perturbation} (also called \textit{sanitization}) which encompasses a
general class of database modification techniques that ensure that a user only
interacts with a modified database that is derived from the original (e.g.:
\cite{Dalenius,Sweeney,Ag_Ag,Chawla01}). Most of the current
perturbation-based approaches are heuristic and application-specific and often
focus on additive noise approaches.

Perturbation techniques depend on whether the database is considered
\textit{interactive} (i.e. whether the user can issue more queries after
seeing earlier responses) or \textit{non-interactive} \cite{Dwork_DP}. In the
non-interactive model, the database is published after a \textit{sanitization}
process in which personal identifiers are eliminated and the data is perturbed
using one of many possible input perturbation approaches; alternately in the
interactive model, the database adds noise to the response based on a data model.

In order to quantify the privacy and utility afforded by a data source,
metrics are critical. The concept of $k$-\textit{anonymity} proposed by
Sweeney \cite{Sweeney} captures the intuitive notion of privacy that every
individual entry should be indistinguishable from $(k-1)$ other entries for
some large value of $k$. More recently, researchers in the data mining
community have proposed to quantify the privacy loss resulting from data
disclosure as the mutual information between attribute values in the original
and perturbed data sets, both modeled as random variables \cite{Ag_Ag}.
Finally, motivated by cryptographic models, the concept of \emph{differential
privacy} from theoretical computer science \cite{Dwork_DP,Dwork_DP_Survey} has
created a universal model for privacy which measures the risk of loss of
privacy to an individual whose data is in a statistical database. However,
this work as well the others described above do not propose a companion
universal utility metric that can be guaranteed along with privacy.

\subsection{Privacy vs. Utility}

While the privacy problem has been studied by multiple communities using
multiple approaches, the companion utility problem has not been studied as
analytically and exhaustively except in the context of specific applications.
Indeed, most discussions of privacy assume an implicit utility that is left
unstated or unmeasured. Utility of a data source is, by necessity, a relative
concept and is measured from the point of view of the user: utility is maximal
when the user gets full information flow and reduces when the flow of certain
information is reduced either by restriction or the addition of noise. The
general concept of utility as a measure of the approximation to an underlying
(but undisclosed) quantity is a fertile area of research (e.g.:
\cite{Util3,Tuncel1}). However, these measures have not been customized for
the context of privacy enhancement. Heuristic measures of utility in the
context of privacy have been proposed (e.g.: \cite{Chawla01}) but they do not
yield a general notion of utility. In our proposed work, we will use a working
definition of utility as the measure of the \textit{distance} or
\textit{divergence} (using suitably chosen metrics such as Euclidean or
Kullback-Leibler divergence) between the original and sanitized databases.

\section{An Information-Theoretic Approach}

\subsection{\label{Sec_DB_Model}Model for Databases}

\textit{Circumventing the semantic issue}: In general, utility and privacy
metrics tend to be application specific. Focusing our efforts on developing an
analytical model, we propose to capture a canonical database model and
representative abstract metrics. Such a model will circumvent the classic
privacy issues related to the semantics of the data by assuming that there
exist forward and reverse maps of the data set to the proposed abstract format
(for e.g., a string of bits or a sequence of real values). Such mappings are
often implicitly assumed in the privacy literature
\cite{Chawla01,Ag_Ag,Dwork_DP}; our motivation for making it explicit is to
separate the semantic issues from the abstraction and apply Shannon-theoretic techniques.

\textit{Model}: Our proposed model focuses on large databases with $K$
attributes per entry. Let $X_{k}\in\mathcal{X}_{k}$ be a random variable
denoting the $k^{th}$ attribute, $k=1,2,\ldots,K,$ and let $\mathbf{X}%
\equiv\left(  X_{1},X_{2},\ldots,X_{K}\right)  $. A database $d$ with $n$ rows
is a sequence of $n$ independent observations of $\mathbf{X}$ with the
distribution%
\begin{equation}
p_{\mathbf{X}}\left(  \mathbf{x}\right)  =p_{X_{1}X_{2}\ldots X_{K}}\left(
x_{1},x_{2},\ldots,x_{K}\right)  \label{Prob_JointDist}%
\end{equation}
which is assumed to be known to the designers of the database. Our assumption
of row independence in (\ref{Prob_JointDist}) is justified because correlation
in databases is typically across attributes and not across entries. We write
$\mathbf{X}^{n}=\left(  X_{1}^{n},X_{2}^{n},\ldots,X_{K}^{n}\right)  $ to
denote the $n$ independent observations of $\mathbf{X}$. This database model
is universal in the sense that most practical databases can be mapped to this model.

A joint distribution in (\ref{Prob_JointDist}) models the fact that the
attributes in general are correlated and can reveal information about one
another. In addition to the revealed information, a user of a database can
have access to correlated side information from other information sources. We
model the side-information as an $n$-length sequence $Z^{n}$ which is
correlated with the database entries via a joint distribution $p_{\mathbf{X}%
Z}\left(  \mathbf{x,}z\right)  .$

\textit{Public and private variables}: We consider a general model in which
some attributes need to be kept private while the source can reveal a function
of some or all of the attributes. We write $\mathcal{K}_{r}$ and
$\mathcal{K}_{h}$ to denote sets of private (subscript $h$ for hidden) and
public (subscript $r$ for revealed) attributes, respectively, such that
$\mathcal{K}_{r}\cup\mathcal{K}_{h}=\mathcal{K\equiv}\left\{  1,2,\ldots
,K\right\}  $. We further denote the corresponding collections of public and
private attributes by $\mathbf{X}_{r}\equiv\left\{  X_{k}\right\}
_{k\in\mathcal{K}_{r}}$ and $\mathbf{X}_{h}\equiv\left\{  X_{k}\right\}
_{k\in\mathcal{K}_{h}}$, respectively. Our notation allows for an attribute to
be both public and private; this is to account for the fact that a database
may need to reveal a function of an attribute while keeping the attribute
itself private. In general, a database can choose to keep public (or private)
one or more attributes ($K>1)$. Irrespective of the number of private
attributes, a non-zero utility results only when the database reveals an
appropriate function of some or all of its attributes.

\textit{Special cases}: For $K=1$, the lone attribute of each entry (row) is
both public and private, and thus, we have $X\equiv X_{r}\equiv X_{h}$. Such a
model is appropriate for data mining \cite{Ag_Ag} and census
\cite{Dalenius,Chawla01} data sets in which utility generally is achieved by
revealing a function of every entry of the database while simultaneously
ensuring that no entry is completely revealed. For $K=2$ and $\mathcal{K}%
_{h}\cup\mathcal{K}_{r}=\mathcal{K}$ and $\mathcal{K}_{h}\cap\mathcal{K}%
_{r}=\emptyset,$ we obtain the Yamamoto model in \cite{Yamamoto}.

\subsection{Metrics: The Privacy and Utility Principle}

Even though utility and privacy measures tend to be specific to the
application, there is a fundamental principle that unifies all these measures
in the abstract domain. The aim of a privacy-preserving database is to provide
some measure of utility to the user while at the same time guaranteeing a
measure of privacy for the entries in the database.

A user perceives the utility of a perturbed database to be high as long as the
response is similar to the response of the original database; thus, the
utility is highest of an original (unpertubed) database and goes to zero when
the perturbed database is completely unrelated to the original database.
Accordingly, our utility metric is an appropriately chosen average `distance'
function between the original and the perturbed databases. Privacy, on the
other hand, is maximized when the perturbed response is completely independent
of the data. Our privacy metric measures the difficulty of extracting any
private information from the response, i.e., the amount of uncertainty or
\textit{equivocation }about the private attributes given the response.

\subsection{\label{SS2}A Privacy-Utility Tradeoff Model}

We now propose a privacy-utility model for databases. \textit{Our primary
contribution is demonstrating the equivalence between the database privacy
problem and a source coding problem with additional privacy constraints}. For
our abstract universal database model, sanitization is thus a problem of
mapping a set of database entries to a different set subject to specific
utility and privacy requirements. Our notation below relies on this abstraction.\ 

Recall that a database $d$ with $n$ rows is an instantiation of $\mathbf{X}%
^{n}$. Thus, we will henceforth refer to a real database $d$ as an
\textit{input sequence} and to the corresponding sanitized database (SDB)
$d^{\prime}$ as an \textit{output sequence}. When the user has access to side
information, the \textit{reconstructed sequence} at the user will in general
be different from the SDB sequence.

Our coding scheme consists of an encoder $F_{E}$ which is a mapping from the
set of all input sequences (i.e., all databases $d$ picked from an underlying
distribution$)$ to a set of indices~$\mathcal{W}\equiv\left\{  1,2,\ldots
,M\right\}  $ and an associated table of output sequences (each of which is a
$d^{\prime})$ with a one-to-one mapping to the set of indices given by%
\begin{equation}
F_{E}:\left(  \mathcal{X}_{1}^{n}\times\mathcal{X}_{2}^{n}\times\ldots
\times\mathcal{X}_{k}^{n}\right)  _{k\in\mathcal{K}_{enc}}\rightarrow
\mathcal{W}\equiv\left\{  SDB_{k}\right\}  _{k=1}^{M} \label{F_Enc}%
\end{equation}
where $\mathcal{K}_{r}\subseteq\mathcal{K}_{enc}\subseteq\mathcal{K}$ and
$M=2^{nR}$ is the number of output (sanitized) sequences created from the set
of all input sequences. The encoding rate $R$ is the number of bits per entry
(without loss of generality, we assume $n$ entries in $d$ and $d^{\prime}$) of
the sanitized database. The encoding $F_{E}$ in (\ref{F_Enc}) includes both
public and private attributes in order to model the general case in which the
sanitization depends on a subset of all attributes.

A user with a view of the SDB (i.e., an index $w\in\mathcal{W}$ for every $d)$
and with access to side information $Z^{n}$, whose entries $Z_{i}$,
$i=1,2,\ldots,n,$ take values in the alphabet $\mathcal{Z}$, reconstructs the
database $d^{\prime}$ via the mapping%
\begin{equation}
F_{D}:\mathcal{W}\times\mathcal{Z}^{n}\rightarrow\left\{  \mathbf{\hat{x}%
}_{r,m}^{n}\right\}  _{m=1}^{M}\in\left(
%TCIMACRO{\tprod \nolimits_{k\in\mathcal{K}_{r}}}%
%BeginExpansion
{\textstyle\prod\nolimits_{k\in\mathcal{K}_{r}}}
%EndExpansion
\mathcal{\hat{X}}_{k}^{n}\right)  \label{F_Dec}%
\end{equation}
where $\mathbf{\hat{X}}_{r}^{n}=F_{D}\left(  F_{E}\left(  \mathbf{X}%
^{n}\right)  \right)  $.

A database may need to satisfy multiple utility constraints for different
(disjoint) subsets of attributes, and thus, we consider a general framework
with $L\geq1$ utility functions that need to be satisfied. Relying on the
distance based utility principle, we model the $l^{th}$ utility,
$l=1,2,\ldots,L,$ via the requirement that the average \textit{distortion}
$\Delta_{l}$ of a function $f_{l}$ of the revealed variables is upper bounded,
for some $\epsilon>0$, as%
\begin{multline}
u_{l}:\Delta_{l}\equiv\mathbb{E}\left[  \frac{1}{n}%
%TCIMACRO{\tsum _{i=1}^{n}}%
%BeginExpansion
{\textstyle\sum_{i=1}^{n}}
%EndExpansion
g\left(  f_{l}\left(  \mathbf{X}_{r,i}\right)  ,f_{l}\left(  \mathbf{\hat{X}%
}_{r,i}\right)  \right)  \right]  \leq D_{l}+\epsilon\text{, }%
\label{Utility_mod}\\
l=1,2,\ldots,L,
\end{multline}
where $g\left(  \cdot,\cdot\right)  $ denotes a distortion function,
$\mathbb{E}$ is the expectation over the joint distribution of $(\mathbf{X}%
_{r},\mathbf{\hat{X}}_{r})$, and the subscript $i$ in $\mathbf{X}_{r,i}$ and
$\mathbf{\hat{X}}_{r,i}$ denotes the $i^{th}$ entry of $\mathbf{X}_{r}^{n}$
and $\mathbf{\hat{X}}_{r}^{n}$, respectively. Examples of distance-based
distortion functions include the Euclidean distance for Gaussian distributed
database entries, the Hamming distance for binary input and output sequences,
and the Kullback-Leibler (K-L) `distance' comparing the input and output distributions.

Having argued that a quantifiable uncertainty captures the privacy of a
database, we model the uncertainty or equivocation about the private variables
using the entropy function as
\begin{equation}
p:\Delta_{p}\equiv\frac{1}{n}H\left(  \mathbf{X}_{h}^{n}|W,Z^{n}\right)  \geq
E-\epsilon, \label{Equivoc}%
\end{equation}
i.e., we require the average number of uncertain bits per dimension to be
lower bounded by $E$. The case in which side information is not available at
the user is obtained by simply setting $Z^{n}=0$ in (\ref{F_Dec}) and
(\ref{Equivoc}). While our general problem allows separate constraints on the
privacy and utility, we show later that for specific canonical databases
(census and data mining) \textit{a constraint on only one} of them (utility or
privacy) suffices (see Corollary \ref{Corr_1} in Section \ref{Sec_SI}).

The utility and privacy metrics in (\ref{Utility_mod}) and (\ref{Equivoc}),
respectively, capture two aspects of our universal model: a) both represent
averages by computing the metrics across all database instantiations $d$, and
b) the metrics bound the average distortion and privacy per entry. Thus, as
the likelihood of the non-typical sequences decreases exponentially with
increasing $n$ (very large databases), these guarantees apply nearly uniformly
to all (typical)\ entries. Our general model also encompasses the fact that
the exact mapping from the distortion and equivocation domains to the utility
and privacy domains, respectively, can depend on the application domain. We
write $D\equiv(D_{1},D_{2},\ldots,$ $D_{L})$ and $\Delta\equiv(\Delta
_{1},\Delta_{2},\ldots,\Delta_{L})$. Based on our notation thus far, we define
the utility-privacy tradeoff region as follows.

\begin{definition}
The utility-privacy tradeoff region $\mathcal{T}$ is the set of all feasible
utility-privacy tuples $(D,E)$ for which there exists a coding scheme $\left(
F_{E},F_{D}\right)  $ given by (\ref{F_Enc}) and (\ref{F_Dec}), respectively,
with parameters $(n,M,\Delta,\Delta_{p})$ satisfying the constraints in
(\ref{Utility_mod}) and (\ref{Equivoc}).
\end{definition}

\subsection{\label{Sec_RDE}Equivalence of Utility-Privacy and
Rate-Distortion-Equivocation}

We now present an argument for the equivalence of the above utility-privacy
tradeoff analysis with a rate-distortion-equivocation analysis of the same
source. For the database source model described here, a classic lossy source
coding problem is defined as follows.

\begin{definition}
The set of tuples $(R,D)$ is said to be feasible (achievable) if there exists
a coding scheme given by (\ref{F_Enc}) and (\ref{F_Dec}) with parameters
$(n,M,\Delta)$ satisfying the constraints in (\ref{Utility_mod}) and a rate
constraint
\begin{equation}
M\leq2^{n\left(  R+\epsilon\right)  }. \label{Rate_constraint}%
\end{equation}

\end{definition}

When an additional privacy constraint in (\ref{Equivoc}) is included, the
source coding problem becomes one of determining the achievable
rate-distortion-equivocation region defined as follows.

\begin{definition}
\label{Def_RDE}The rate-distortion-equivocation region $\mathcal{R}$ is the
set of all tuples $(R,D,E)$ for which there exists a coding scheme given by
(\ref{F_Enc}) and (\ref{F_Dec}) with parameters $(n,M,\Delta,\Delta_{p})$
satisfying the constraints in (\ref{Utility_mod}), (\ref{Equivoc}), and
(\ref{Rate_constraint}). The set of all feasible distortion-equivocation
tuples $\left(  D,E\right)  $ is denoted by $\mathcal{R}_{D-E}$, the
equivocation-distortion function in the $D$-$E$ plane is denoted by
$\Gamma(D)$, and the distortion-equivocation function which quantifies the
rate as a function of both $D$ and $E$ is denoted by $R\left(  D,E\right)  $.
\end{definition}

Thus, a rate-distortion-equivocation code is by definition a (lossy) source
code satisfying a set of distortion constraints that achieves a specific
privacy level for every choice of the distortion tuple. In the following
theorem, we present a basic result capturing the precise relationship between
$\mathcal{T}$ and $\mathcal{R}$. To the best of our knowledge, this is the
first analytical result that quantifies a tight relationship between utility
and privacy. We briefly sketch the proof here; details can be found in
\cite{LS_VP}.

\begin{theorem}
\label{Lemma_equiv}For a database with a set of utility and privacy metrics,
the tightest utility-privacy tradeoff region $\mathcal{T}$ is the
distortion-equivocation region $\mathcal{R}_{D-E}$.
\end{theorem}

\begin{proof}
The crux of our argument is the fact that for any feasible utility level $D$,
choosing the minimum rate $R\left(  D\right)  $, ensures that the least amount
of \textit{information} is revealed about the source via the reconstructed
variables. This in turn ensures that the maximum privacy of the private
attributes is achieved for that utility since, in general, the public and
private variables are correlated. For the same set of utility constraints,
since such a rate requirement is not a part of the utility-privacy model, the
resulting privacy achieved is at most as large as that in $\mathcal{R}_{D-E}$
(see Fig. \ref{Fig_RDE_UP}(a)).
\end{proof}

Implicit in the above argument is the fact that a utility-privacy achieving
code does not perform any better than a rate-distortion-equivocation code in
terms of achieving a lower rate (given by $\log_{2}M/n)$ for the same
distortion and privacy constraints. This is because if such a code exists then
we can always find an equivalent source coding problem for which the code
would violate Shannon's source coding theorem \cite{Shannon_SC}. An immediate
consequence of this is that a distortion-constrained source code suffices to
preserve a desired level of privacy; in other words, \textit{the utility
constraints require revealing data which in turn comes at a certain privacy
cost that must be borne and vice-versa}. We capture this observation in Fig.
\ref{Fig_RDE_UP}(b) where we contrast existing privacy-exclusive and
utility-exclusive regimes (extreme points of the utility-privacy tradeoff
curve) with our more general approach of determining the set of feasible
utility-privacy tradeoff points.%

%TCIMACRO{\TeXButton{B}{\begin{figure*}[tbp] \centering}}%
%BeginExpansion
\begin{figure*}[tbp] \centering
%EndExpansion%
%TCIMACRO{\FRAME{itbpF}{5.4794in}{2.8219in}{0in}{}{}{rde_up_figures.eps}%
%{\special{ language "Scientific Word";  type "GRAPHIC";  display "USEDEF";
%valid_file "F";  width 5.4794in;  height 2.8219in;  depth 0in;
%original-width 17.1977in;  original-height 7.478in;  cropleft "0";
%filename 'RDE_UP_figures.eps';file-properties "XNPEU";}}}%
%BeginExpansion
{\includegraphics[
height=2.8219in,
width=5.4794in
]%
{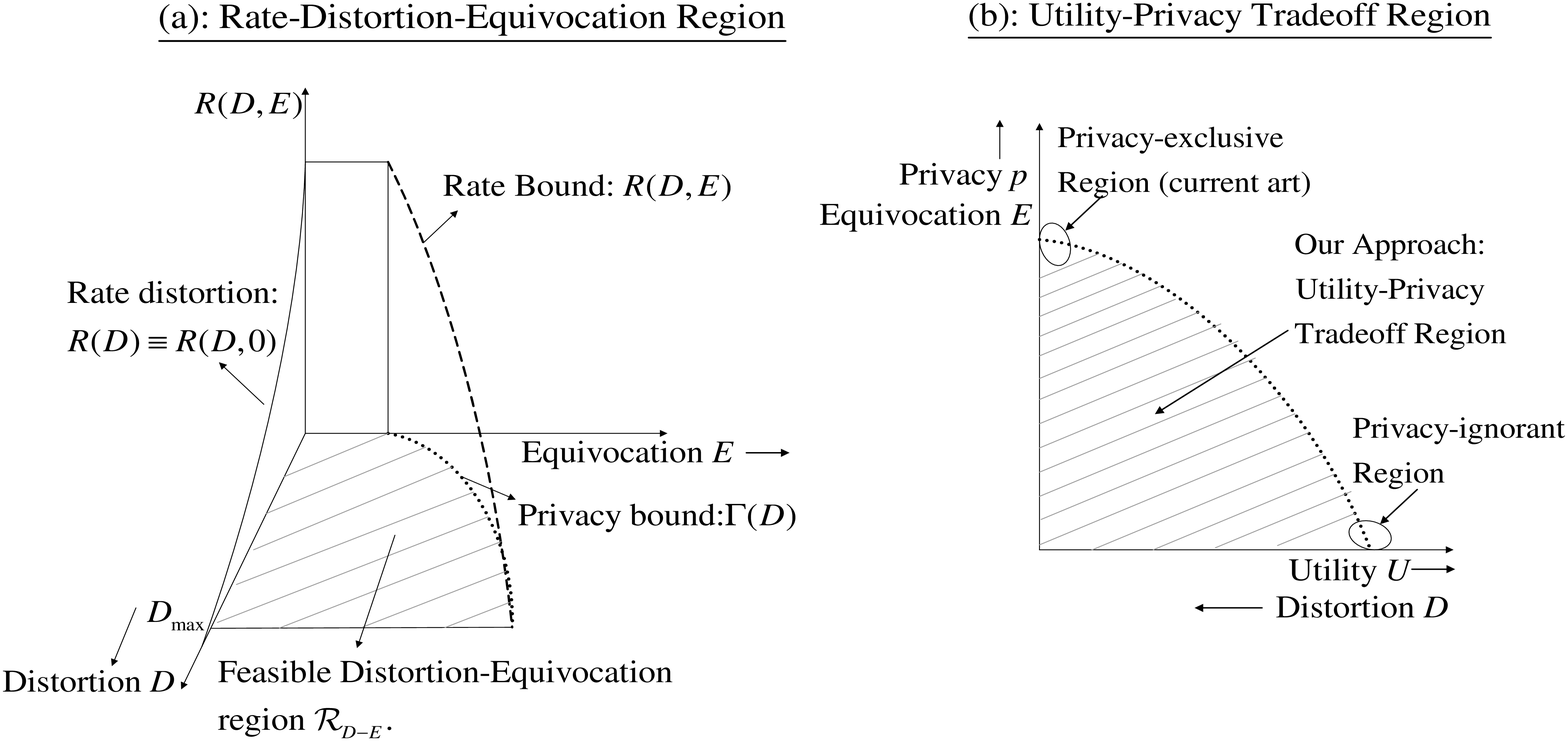}%
}%
%EndExpansion
\caption{(a) Rate Distortion Equivocation Region; (b) Utility-Privacy Tradeoff Region.}\label{Fig_RDE_UP}%
%TCIMACRO{\TeXButton{E}{\end{figure*}}}%
%BeginExpansion
\end{figure*}%
%EndExpansion

From an information-theoretic perspective, the power of Theorem
\ref{Lemma_equiv} is that it allows us to study the larger problem of database
utility-privacy tradeoffs in terms of a relatively familiar problem of source
coding with privacy constraints. As noted previously, this problem has been
studied for a specific source model by Yamamoto and here we expand his elegant
analysis to arbitrary database models including those with side information at
the user.

\subsection{\label{Sec_SI}Capturing the Effects of Side-Information}

It has been illustrated that when a user has access to an external data source
(which is not part of the database under consideration) the level of privacy
that can be guaranteed changes \cite{Sweeney,Dwork_DP}. We cast this problem
in information-theoretic terms as a side information problem.

In an extended version of this work \cite{LS_VP}, we have developed the
tightest utility-privacy tradeoff region for the three cases of a) no side
information ($L=1$ case studied in \cite{Yamamoto}), b) side information only
at the user, and c) side information at both the source (database) and the
user. We present a result for the case with side information at the user only
and for simplicity, we assume a single utility function, i.e., $L=1$. The
proof mimics that of source coding with side information in \cite{Wyner_Ziv}
and therefore, involves the use of an auxiliary random variable $U$. The proof
also includes bounds on the equivocation along the lines of those in
\cite[Appendix 1]{Yamamoto}. The following theorem defines the bounds on the
region $\mathcal{R}$ in Definition \ref{Def_RDE} via the functions $\Gamma(D)$
and $R(D,E)$ where $\Gamma(D)$ bounds the maximal achievable privacy and
$R(D,E)$ is the minimal information rate (see Fig. \ref{Fig_RDE_UP}(a)) for
very large databases $\left(  n\rightarrow\infty\right)  $. The proof is
omitted due to space and can be found in \cite{LS_VP}.

\begin{theorem}
\label{Pro_Prop1}For a database with side information available only at the
user, the functions $\Gamma(D)$ and $R\left(  D,E\right)  $ and the regions
$\mathcal{R}_{D-E}$ and $\mathcal{R}$ are given by%
\begin{align}
\Gamma\left(  D\right)   &  =\sup_{p\left(  \mathbf{x}_{r},\mathbf{x}%
_{h}\right)  p\left(  u|\mathbf{x}_{r},\mathbf{x}_{h}\right)  \in
\mathcal{P}\left(  D\right)  }H(\mathbf{X}_{h}|UZ)\label{TD_SIu}\\
R\left(  D,E\right)   &  =\inf_{p\left(  \mathbf{x}_{r},\mathbf{x}_{h}\right)
p\left(  u|\mathbf{x}_{r},\mathbf{x}_{h}\right)  \in\mathcal{P}\left(
D,E\right)  }I(\mathbf{X}_{h}\mathbf{X}_{r};U)-I(Z;U) \label{RDE_SIu}%
\end{align}%
\begin{equation}
\mathcal{R}_{D-E}=\left\{  \left(  D,E\right)  :D\geq0,0\leq E\leq
\Gamma\left(  D\right)  \right\}  \label{RDEreg_SIu}%
\end{equation}%
\begin{equation}
\mathcal{R}=\left\{  \left(  R,D,E\right)  :D\geq0,0\leq E\leq\Gamma\left(
D\right)  ,R\geq R\left(  D,E\right)  \right\}  \label{Rreg_SIu}%
\end{equation}
where $\mathcal{P}\left(  D,E\right)  $ is the set of all $p(\mathbf{x}%
_{r},\mathbf{x}_{h},z)p(u|\mathbf{x}_{r},\mathbf{x}_{h})$ such that
$\mathbb{E}\left[  d\left(  \mathbf{X}_{r},g\left(  U,Z\right)  \right)
\right]  \leq D$ and $H(\mathbf{X}_{h}|UZ)\geq E$ while $\mathcal{P}\left(
D\right)  $ is defined as%
\begin{equation}
\mathcal{P}\left(  D\right)  \equiv%
%TCIMACRO{\tbigcup _{H(\mathbf{X}_{h}|\mathbf{X}_{r}Z)\leq E\leq H(\mathbf{X}%
%_{h}|Z)}}%
%BeginExpansion
{\textstyle\bigcup_{H(\mathbf{X}_{h}|\mathbf{X}_{r}Z)\leq E\leq H(\mathbf{X}%
_{h}|Z)}}
%EndExpansion
\mathcal{P}\left(  D,E\right)  .
\end{equation}

\end{theorem}

While Theorem \ref{Pro_Prop1} applies to a variety of database models, it is
extremely useful in quantifying the utility-privacy tradeoff for the following
special cases of interest.

i) \textit{The single database problem} (i.e., no side information\textit{)}:
\textit{SDB\ is revealed}. Here, we have $Z=0$ and $U=\hat{X}_{r}$, i.e., the
reconstructed vectors seen by the user are the same as the SDB vectors.

ii) \textit{Completely hidden private variables}: \textit{Privacy is
completely a function of the statistical relationship between public, private,
and side information data}. The expression for $R(D,E)$ in (\ref{RDE_SIu})
assumes the most general model of encoding both the private and the public
variables. When the private variables can only be deduced from the revealed
variables, i.e., $\mathbf{X}_{h}-\mathbf{X}_{r}-U$ is a Markov chain, the
expression for $R(D,E)$ in (\ref{RDE_SIu}) will simplify to the Wyner-Ziv
source coding formulation \cite{Wyner_Ziv}, thus clearly demonstrating that
the privacy of the hidden variables is a function of both the correlation
between the hidden and revealed variables and the distortion constraint.

iii)\ \textit{Census and data mining problems without side information}:
\textit{Information rate completely determines privacy achievable}. For $Z=0$,
setting $\mathbf{X}_{r}=\mathbf{X}_{h}\equiv X$ (such that $U=\hat{X}$), we
obtain the census/data mining problem discussed earlier. With this
substitution, from Theorem \ref{Pro_Prop1}, we have the maximal achievable
equivocation $\Gamma(D)=H(X)-R(D),$ where now $R(D)\equiv R(D,E)$. Our
analysis formalizes the intuition in \cite{Ag_Ag} for using the mutual
information as an estimate of the privacy lost. However in contrast to
\cite{Ag_Ag} in which the underlying perturbation model is an additive noise
model, we assume a perturbation model most appropriate for the input
statistics, i.e., the stochastic relationship between the output and input
variables is chosen to minimize the rate of information transfer. This
fundamental result is captured in the following corollary.

\begin{corollary}
\label{Corr_1}For the special case of $K=1$, i.e., $\mathbf{X}_{r}%
=\mathbf{X}_{h}\equiv X$, the utility-privacy problem is completely defined by
a utility constraint since the maximum achievable equivocation is directly
obtainable from the minimal information transfer rate.
\end{corollary}

\subsection{A Successive Disclosure Problem}

As mentioned earlier, databases can be broadly categorized as non-interactive
and interactive depending on whether the data is sanitized once before
publishing or repeatedly in response to each query, respectively. For census
and similar statistical databases a one-shot sanitization is typical whereas
for more interactive databases multiple queries can lead to multiple sanitizations.

\textit{Single-query model}: The model and analysis proposed in Sections
\ref{Sec_DB_Model}-\ref{Sec_RDE} capture the non-interactive database model
and the resulting utility-privacy tradeoff region. For this one-shot model,
sanitization is determined by the choice of the utility and privacy metrics
defined \textit{a priori}. In contrast to existing approaches that are
dominantly focused on additive noise perturbations satisfying a large set of
queries \cite{DinNis,Dwo_NoiSen}, our one-shot approach is independent of
queries and is designed to satisfy specific utility and privacy constraints.
Such a model is relevant for databases such as those with medical and clinical
data that may find repeated uses in the future but with queries that cannot be
predicted ahead of time or which require query-independent strict sanitization
prior to interaction to ensure regulatory compliance (e.g., US\ HIPAA privacy
policies \cite{HIPAA}).

\textit{Multiple-query model}: For a large majority of data repositories,
utility is a function of their usage and as such the problem of addressing the
utility-privacy tradeoffs in a multiple query model is imperative. A
side-effect of allowing multiple queries is that a user can refine her query
to learn more information at each step, which in turn can lead to privacy
breaches. Our aim is to determine if a certain level of overall utility can be
guaranteed while preserving a desired overall privacy threshold. In the
absence of disclosure controls, a database will typically respond to each
query independently of the previous queries. We seek to develop a model in
which the database is cognizant of current and past queries in responding to
future queries. To this end, we assume the existence of a data collector that
provides an interface for the user to submit queries and collate the responses
over multiple queries, a common assumption in the multi-query literature
\cite{Dwork_DP,Dwo_NoiSen,Dwork_DP_Survey}. For this model, under the
assumption that the user wishes to obtain a refined view of the source, we
propose to \textit{determine whether a source can be successively disclosed},
i.e., whether a set of overall utility and privacy constraints can be
satisfied via multiple disclosures with increasing refinement at each stage
and without any information loss relative to an equivalent single-shot model
with the same overall utility and privacy constraints.

This problem of successive disclosure has a natural relationship to a problem
of \textit{successive refinement} in information theory, which pertains to
determining whether successively revealing data from a source with decreasing
distortion at each stage can ensure no rate loss relative to a one-shot
approach with the same final distortion \cite{CovEq,Rimoldi,Ahlswede}. We
demonstrate this analogy in Fig. \ref{Fig_SR_SD} where, at the first stage,
the user obtains a specific view (denoted $\hat{X}_{1}$ of a source $X)$ of
the source which in conjunction with the second stage provides a final refined
view $\hat{X}_{2}$. While the successive refinement problem is to determine
whether $R_{2}=R\left(  D_{2}\right)  $, the successive disclosure problem is
that of determining whether $R_{2}=R\left(  D_{2},E_{2}\right)  $ where
$D_{2}<D_{1}$ and $E_{2}<E_{1}$. As with the successive refinement problem,
our results can help determine the conditions and relationships between the
input and output sequences under which a source can be disclosed successively.%

%TCIMACRO{\TeXButton{B}{\begin{figure*}[tbp] \centering}}%
%BeginExpansion
\begin{figure*}[tbp] \centering
%EndExpansion%
%TCIMACRO{\FRAME{itbpF}{6.3105in}{2.3263in}{0in}{}{}{sr_sd.eps}%
%width 6.3105in;  height 2.3263in;  depth 0in;  original-width 16.1729in;
%original-height 5.9222in;  cropleft "0";  croptop "1";  cropright "1";
%cropbottom "0";  filename 'SR_SD.eps';file-properties "XNPEU";}}}%
%BeginExpansion
{\includegraphics[
height=2.3263in,
width=6.3105in
]%
{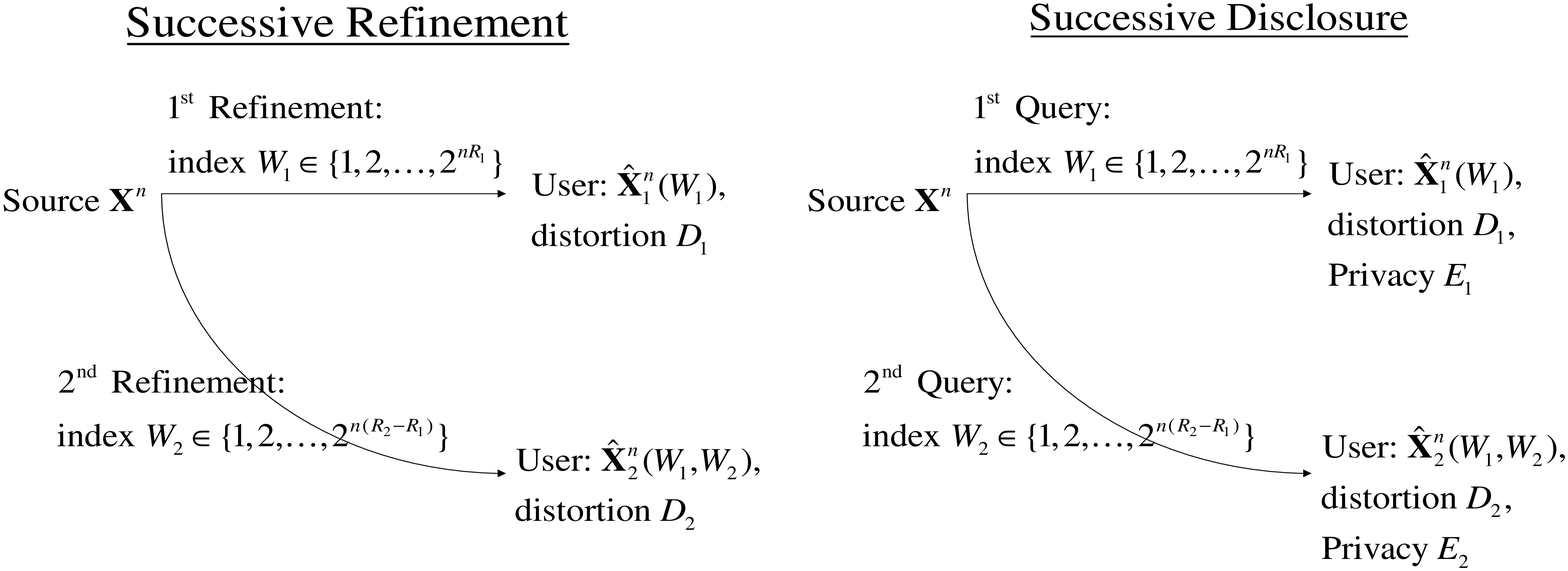}%
}%
%EndExpansion
\caption{Successive Refinement and Successive Disclosure Problems.}\label{Fig_SR_SD}%
%TCIMACRO{\TeXButton{E}{\end{figure*}}}%
%BeginExpansion
\end{figure*}%
%EndExpansion

Analogous to successive refinement, we start by studying a \textit{multiple
disclosure} problem in which we seek to determine the rates $R_{0}$ and
$R_{1}$ at which the database responds with distortion (utility) and privacy
levels $\left(  D_{0},E_{0}\right)  $ and $\left(  D_{1},E_{1}\right)  $ to
two queries, respectively, such that a user using both query responses can
reconstruct a response at a distortion-privacy level of $\left(  D_{2}%
,E_{2}\right)  $. Analogous to the relationship between multiple description
and successive refinement, the successive disclosure problem described here is
a special case of the multiple disclosure problem for which there is no rate
loss, i.e., $R_{1}=R\left(  D_{1},E_{1}\right)  $ and $R_{0}+R_{1}=R\left(
D_{2},E_{2}\right)  $.

While a detailed analysis of this problem can be found in an extended version
of this work \cite{LS_VP}, we now present two example privacy problems for
which the successive refinement problem presents immediate insights on the
effects of refined disclosure. The two problems are privacy preservation in
census and data mining databases, and in both cases, we briefly argue that the
successive disclosure problem simplifies to the successive refinement problem.
Recall that in Corollary \ref{Corr_1}, we showed that the census and data
mining problems are special cases for which the rate-distortion-equivocation
region is directly obtainable from the rate-distortion curve because for both
problems the public and the private variables are the same as a result of
which the maximum achievable equivocation is directly obtainable from the
rate-distortion function. The following theorem summarizes our result.

\begin{theorem}
For $K=1$ databases, successive disclosure with distortion-privacy pairs
$\left(  D_{1},E_{1}\right)  $ and $\left(  D_{2},E_{2}\right)  $ are
achievable if and only if there exists a conditional distribution $p\left(
\hat{x}_{1},\hat{x}_{2}|x\right)  $ with%
\begin{equation}%
\begin{array}
[c]{cc}%
\mathbb{E}\left[  g\left(  X,\hat{X}_{k}\right)  \right]  \leq D_{k}, & k=1,2,
\end{array}
\end{equation}
such that%
\begin{equation}%
\begin{array}
[c]{cc}%
R\left(  D_{k},E_{k}\right)  =I(X;\hat{X}_{k}), & k=1,2,
\end{array}
\end{equation}
and $X-\hat{X}_{2}-\hat{X}_{1}$ form a Markov chain, i.e.,%
\begin{equation}
p\left(  \hat{x}_{1},\hat{x}_{2}|x\right)  =p\left(  \hat{x}_{2}|x\right)
p\left(  \hat{x}_{1}|\hat{x}_{2}\right)  .
\end{equation}

\end{theorem}

Thus, for these two special but fundamentally important problems, we can show
that the Markov condition $X-\hat{X}_{2}-\hat{X}_{1}$ (see Fig.
\ref{Fig_SR_SD}) required for successive refinement \cite[Theorem 2]{CovEq}
also hold here and in fact suffices to satisfy the successive disclosure
requirement of no additional rate or privacy leakage. More work is needed to
address questions such as the practical implications of the above Markov
condition \cite{Rimoldi} and generalizing the solution to arbitrary sources.

\section{Concluding Remarks}

We have presented an abstract model for databases with an arbitrary number of
public and private variables, developed application-independent privacy and
utility metrics, used rate distortion theory to determine the fundamental
utility-privacy tradeoff limits, and introduced a successive disclosure
problem to study utility-privacy tradeoffs and determine the conditions for no
privacy loss for multiple query data sources. Future work includes
generalizing the results to distributed data sources and relating current
approaches in computer science and our universal approach.

\bibliographystyle{IEEEtran}
\bibliography{DB_refs}

\end{document}